\documentclass[submission,copyright,creativecommons]{eptcs}
\providecommand{\event}{LFMTP 2024} 

\usepackage{iftex}

\ifpdf
  \usepackage{underscore}         
  \usepackage[T1]{fontenc}        
\else
  \usepackage{breakurl}           
\fi

\title{Proofs for Free in the $\lambda\Pi$-Calculus Modulo Theory}
\author{Thomas Traversié
\institute{Université Paris-Saclay, CentraleSupélec, MICS \\ Gif-sur-Yvette, France}
\institute{Université Paris-Saclay, Inria, CNRS, ENS-Paris-Saclay, LMF \\ Gif-sur-Yvette, France}
\email{thomas.traversie@centralesupelec.fr}}
\def\titlerunning{Proofs for Free in the $\lambda\Pi$-Calculus Modulo Theory}
\def\authorrunning{T. Traversié}

\usepackage{xspace}
\usepackage{amsmath}
\usepackage{amssymb}
\usepackage{amsthm}
\usepackage{mathpartir}
\usepackage{mathtools}
\usepackage[capitalise]{cleveref}
\usepackage{stmaryrd}
\usepackage{enumitem}

\newtheorem{definition}{Definition}
\newtheorem{lemma}{Lemma}
\newtheorem{proposition}{Proposition}
\newtheorem{theorem}{Theorem}
\newtheorem{example}{Example}

\def\ra{\rightarrow}
\def\lra{\hookrightarrow}

\def\T{\mathbb{T}}
\def\S{\mathbb{S}}

\def\Type{\mbox{\tt TYPE}}
\def\Kind{\mbox{\tt KIND}}
\def\Set{{\it Set}}
\def\El{{\it El}}

\def\Prf{{\it Prf}}
\def\imp{\mathbin{\Rightarrow}}
\def\impd{\mathbin{\Rightarrow_d}}
\def\arr{\mathbin{\rightsquigarrow}}
\def\arrd{\mathbin{\rightsquigarrow_d}}
\def\o{o}
\def\blpi{\pi}
\def\fa{\forall}

\def\nat{\mathsf{nat}}
\def\int{\mathsf{int}}
\def\succ{\mathsf{succ}}
\def\pred{\mathsf{pred}}
\def\rec{\mathsf{rec}}
\def\ax{\mathsf{ax}}
\def\thm{\mathsf{thm}}

\def\graph{\mathsf{graph}}
\def\node{\mathsf{node}}
\def\root{\mathsf{root}}
\def\et{\mathsf{eta}}
\def\chgr{\mathsf{cr}}

\def\ts#1{#1^*}
\def\tp#1{#1^+}
\def\tsp#1{#1^{*,+}}

\def\sth{\S} 
\def\tth{\T} 

\def\al{\textit{et al.}}

\newcommand\lampi{{\sc Lambdapi}\xspace}
\newcommand\dk{{\sc Dedukti}\xspace}

\newcommand{\lpc}{$\lambda \Pi$-calculus\xspace}
\newcommand{\lpcm}{$\lambda\Pi$-calculus modulo theory\xspace}

\allowdisplaybreaks

\begin{document}

\maketitle

\begin{abstract}
Parametricity allows the transfer of proofs between different implementations of the same data structure. The $\lambda\Pi$-calculus modulo theory is an extension of the $\lambda$-calculus with dependent types and user-defined rewrite rules. It is a logical framework, used to exchange proofs between different proof systems. We define an interpretation of theories of the $\lambda\Pi$-calculus modulo theory, inspired by parametricity. Such an interpretation allows to transfer proofs for free between theories that feature the notions of proposition and proof, when the source theory can be embedded into the target theory.
\end{abstract}

\section{Introduction}

Many proof assistants have been developed during the past decades, such as \textsc{Agda}, \textsc{Coq}, \textsc{HOL Light}, \textsc{Isabelle}, \textsc{Lean} or \textsc{Mizar}. All those systems have their own theoretical foundations and proof language. If a library of proofs has been formalized in some proof assistant, one would ideally like to export it automatically to any other proof assistant. That is why the question of the \textit{interoperability} between proof systems arises. Exchanging formal proofs between different proof systems strengthen re-usability, re-checking and preservation of libraries. For this purpose, Cousineau and Dowek developed the \lpcm~\cite{lambdapi}, that combines $\lambda$-calculus with dependent types and user-defined rewrite rules. It is a logical framework, in which theories are defined by typed constants and rewrite rules, specified by the users. Many theories can be expressed in the \lpcm~\cite{theoryU}, such as Predicate Logic, Simple Type Theory and the Calculus of Constructions. Most of all, theories from various proof assistants can be expressed in this logical framework. As a consequence, it can be used as a common framework for exchanging proofs between proof systems~\cite{thire}. The \lpcm has been implemented in the concrete language \dk~\cite{expressing,deduktiengine} and in the \lampi proof assistant, which features user-friendly proof tactics.

The problem of the exchange of proofs also emerges when it comes to the different implementations of a same data structure. One would like to share the theorems proved for one implementation to all the other implementations of the same data structure, without additional efforts. One method to derive theorems for free is to use \textit{parametricity}. Reynolds~\cite{reynolds} originally introduced an abstraction theorem, stating that the different implementations of a polymorphic function behave similarly. Wadler~\cite{wadler} used this result to derive properties satisfied by polymorphic functions, depending on their types. In other words, all functions of the same abstract type satisfy the same theorems. Bernardy \al~\cite{parametricity_conf,parametricity_journal} later extended parametricity to Pure Type Systems. Keller and Lasson~\cite{keller_lasson} investigated parametricity for the Calculus of Inductive Constructions, the language behind the \textsc{Coq} proof assistant. More recently, Cohen \al~\cite{trocq} developed a parametricity framework and implemented \textsc{Trocq}, a \textsc{Coq} plugin for proof transfer based on parametricity. The exchange of proofs---the very purpose of the \lpcm---is therefore an important application of the parametricity translations.

Transferring databases of proofs is relevant when working with related mathematical structures. For instance, if we have proved theorems in a theory of natural numbers and we want to use them in a theory of integers, we would like to export the proofs for non-negative integers. The same issue arises concerning various mathematical structures and databases of proofs, as we can \textit{embed} natural numbers into reals, reals into reals extended with infinity elements, or sets into pointed graphs~\cite{deduktiz}. It would therefore be interesting to exchange proofs between theories of the \lpcm, when the source theory can be embedded into the target theory.

\paragraph{Contribution.} In this paper, we define an interpretation of theories of the \lpcm, when they feature a prelude encoding of the notions of proposition and proof. Such an interpretation, inspired by parametricity, applies when we can embed the source theory $\sth$ into the target theory $\tth$. The interpretation depends on parameters, given by the user for representing each constant of the source theory by a term in the target theory. We provide the parameters necessary for interpreting the prelude encoding. We show that if $\sth$ has an interpretation in $\tth$, then the proofs written inside $\sth$ can be transformed into proofs written inside $\tth$. This interpretation comes with a relative consistency theorem: if $\tth$ is consistent, then $\sth$ is consistent too. 

In order to illustrate this interpretation, we embed a theory of natural numbers into a theory of integers. This example, as well as the parameters for the prelude encoding, are given in \dk, and are available at \url{https://github.com/thomastraversie/InterpDK}.

\paragraph{Outline of the paper.} In \Cref{sec_lpcm}, we give a formal presentation of the \lpcm, and we detail a prelude encoding of the notions of proposition and proof. In \Cref{sec_interp}, we define an interpretation of theories of the \lpcm. In particular, we specify the parameters required for interpreting the prelude encoding. We prove the interpretation theorem and the relative consistency theorem. At the end, we show how this interpretation can be used to derive theorems for free, taking the running example of natural numbers and integers.

\section{Theories in the \texorpdfstring{$\lambda\Pi$}{lambdaPi}-Calculus Modulo Theory}
\label{sec_lpcm}

In this section, we give a formal definition of the syntax and type system of the \lpcm. We present a standard way of expressing the notions of proposition and proof in it---called prelude encoding---and we emphasize the theories that will be considered in the rest of the paper.

\subsection{The \texorpdfstring{$\lambda\Pi$}{lambdaPi}-Calculus Modulo Theory}

The Edinburgh Logical Framework~\cite{LF}, also known as \lpc, is an extension of simply typed $\lambda$-calculus with dependent types. The \lpcm~\cite{lambdapi} is an extension of the Edinburgh Logical Framework, in which user-defined rewrite rules~\cite{rewriteSystem} have been added. Its syntax is given by:
\begin{align*}
&\text{\textit{Sorts}} &s &\Coloneqq \Type ~\mid~ \Kind \\
&\text{\textit{Terms}} &t,u, A, B &\Coloneqq c ~\mid~ x ~\mid~ s ~\mid~ \Pi (x : A). ~B ~\mid~ \lambda (x : A). ~t ~\mid~ t ~u \\
&\text{\textit{Contexts}} &\Gamma &\Coloneqq \langle \rangle ~\mid~ \Gamma, x : A \\
&\text{\textit{Signatures}} &\Sigma &\Coloneqq \langle \rangle ~\mid~ \Sigma, c : A ~\mid~ \Sigma, \ell \lra r
\end{align*}
where $c$ is a constant and $x$ is a variable (ranging over disjoint sets), $\Pi (x : A). ~B$ is a dependent product (simply written $A \ra B$ if $x$ does not occur in $B$), $\lambda (x : A). ~t$ is an abstraction, and $t ~u$ is an application. For convenience, $\lambda (x_1 : A_1). \ldots \lambda (x_n : A_n). ~t$ is written $\lambda (x_1 : A_1) \ldots (x_n : A_n). ~t$ and $\Pi (x_1 : A_1). \ldots \Pi (x_n : A_n). ~B$ is written $\Pi (x_1 : A_1) \ldots (x_n : A_n). ~B$. Terms of type $\Type$ are called types, and terms of type $\Kind$ are called kinds. Signatures and contexts are finite sequences, and are written $\langle \rangle$ when empty. The \lpcm is a logical framework, in which $\Sigma$ is fixed by the users depending on the theory they are working in. Signatures are composed of typed constants $c : A$ (such that $A$ is a closed term, that is a term with no free variables) and rewrite rules $\ell \lra r$ (such that the head-symbol of $\ell$ is a constant). The relation $\lra_{\beta\Sigma}$ is the smallest relation, closed by context, such that if $t$ rewrites to $u$ for some rule in $\Sigma$ or by $\beta$-reduction, then $t \lra_{\beta\Sigma} u$. The conversion $\equiv_{\beta\Sigma}$ is the reflexive, symmetric, and transitive closure of the relation $\lra_{\beta\Sigma}$.

\begin{figure}[ht]
\begin{mathpar}
\inferrule*[right={[Empty]}]{ }{\vdash \langle \rangle}

\inferrule*[right={[Decl] $x \notin \Gamma$}]{\vdash \Gamma \\ \Gamma \vdash A : \Type}{\vdash \Gamma, x : A}

\inferrule*[right={[Sort]}]{\vdash \Gamma}{\Gamma \vdash \Type : \Kind}

\inferrule*[right={[Const] $c : A \in \Sigma$}]{\vdash \Gamma \\ \vdash A : s}{\Gamma \vdash c : A}

\inferrule*[right={[Var] $x : A \in \Gamma$}]{\vdash \Gamma}{\Gamma \vdash x : A}

\inferrule*[right={[Prod]}]{\Gamma \vdash A : \Type \\ \Gamma, x : A \vdash B : s}{\Gamma \vdash \Pi (x : A). ~B : s}

\inferrule*[right={[Abs]}]{\Gamma \vdash A : \Type \\ \Gamma, x : A \vdash B : s \\ \Gamma, x : A \vdash t : B}{\Gamma \vdash \lambda (x : A). ~t : \Pi (x : A). ~B}

\inferrule*[right={[App]}]{\Gamma \vdash t : \Pi (x : A). ~B \\ \Gamma \vdash u : A}{\Gamma \vdash t\ u : B[x \leftarrow u]}

\inferrule*[right={[Conv] $A \equiv_{\beta\Sigma} B$}]{\Gamma \vdash t : A \\ \vdash B : s}{\Gamma \vdash t : B}
\end{mathpar}
\caption{Typing rules of the \lpcm.}
\label{typrules_lpcm}
\end{figure}

The judgment $\vdash \Gamma$ means that the context $\Gamma$ is well-formed, and $\Gamma \vdash t : A$ means that $t$ is of type $A$ in the context $\Gamma$. When the context is empty, we simply write $\vdash t : A$. The typing rules for the \lpcm are given in \Cref{typrules_lpcm}. The standard weakening rule is admissible. 

A signature is a theory when its rewrite rules satisfy certain properties. We write $\Lambda_{\Sigma}$ for the set of terms whose constants belong to $\Sigma$. 

\begin{definition}[Theory]
A theory $\T$ in the \lpcm is given by a signature $\Sigma$ such that:
\begin{enumerate}
\item for each rule $\ell \lra r \in \Sigma$, we have $\ell$ and $r$ in $\Lambda_{\Sigma}$,
\item $\lra_{\beta\Sigma}$ is confluent on $\Lambda_{\Sigma}$,
\item for each rule $\ell \lra r \in \Sigma$, for all context $\Gamma$, term $A \in \Lambda_{\Sigma}$ and substitution $\theta$, if $\Gamma \vdash \ell\theta : A$ then $\Gamma \vdash r\theta : A$.
\end{enumerate}
\end{definition}

\begin{lemma}
If $\Gamma \vdash t : A$, then either $A = \Kind$ or $\Gamma \vdash A : s$ for $s = \Type$ or $s = \Kind$. If $\Gamma \vdash \Pi (x : A). ~B : s$, then $\Gamma \vdash A : \Type$.
\end{lemma}

\subsection{A Prelude Encoding}

It is possible to formalize the notions of proposition and proof in the \lpcm~\cite{theoryU}. In particular, this encoding---called prelude encoding---gives the possibility to quantify over certain propositions through codes, which is not possible inside the standard \lpcm. This encoding is defined by the following signature, written $\Sigma_{pre}$.
\begin{flalign*}
&\Set : \Type  & &\o : \Set \\
&\El : \Set \ra \Type & &\Prf : \El ~\o \ra \Type  \\
&\arrd : \Pi (x : \Set). ~(\El ~x \ra \Set) \ra \Set & &\impd : \Pi (x : \El ~\o). ~(\Prf ~x \ra \El ~\o) \ra \El ~\o \\
&\El ~(x \arrd y) \lra \Pi (z : \El ~x). ~\El ~(y ~z) & &\Prf ~(x \impd y) \lra \Pi (z : \Prf ~x). ~\Prf ~(y ~z) \\
&\blpi : \Pi (x : \El ~\o). ~(\Prf ~x \ra \Set) \ra \Set & &\fa : \Pi (x : \Set). ~(\El ~x \ra \El ~\o) \ra \El ~\o \\
&\El ~(\blpi ~x ~y) \lra \Pi (z : \Prf ~x). ~\El ~(y ~z) & &\Prf ~(\fa ~x ~y) \lra \Pi (z : \El ~x). ~\Prf ~(y ~z)
\end{flalign*}
We declare the constant $\Set$, which represents the universe of sorts, along with the injection $\El$ that maps terms of type $\Set$ to the type of its elements. We define a sort $\o$, such that $\El ~\o$ corresponds to the universe of propositions. The injection $\Prf$ maps propositions to the type of its proof. In other words, a term $P$ of type $\El ~\o$ is a proposition, and a term of type $\Prf ~P$ is a proof of $P$. The infix symbol $\arrd$ (respectively $\impd$) is used to represent dependent function types between terms of type $\Set$ (respectively $\El ~\o$). Remark that the symbols $\arrd$ and $\impd$ are generalizations of the usual functionality $\arr$ and implication $\imp$ in the case of dependent types. The symbol $\blpi$ (respectively $\fa$) is used to represent dependent function types between elements of type $\El ~\o$ and $\Set$ (respectively $\Set$ and $\El ~\o$).

While it is not possible to quantify over $\Type$ in the \lpcm, this encoding allows to quantify over propositions---objects of type $\El ~\o$---and then inject them into $\Type$ using $\Prf$. Similarly, we can quantify over sorts---objects of type $\Set$---and then inject them into $\Type$ using $\El$.

\subsection{Theories with Prelude Encoding}

In this paper, we consider theories that feature those basic notions of proposition and proof. More formally, we take theories of the form $\T = \Sigma_{pre} \cup \Sigma_{\T}$, where the user-defined constants $c : A \in \Sigma_{\T}$ have to be expressed in the prelude encoding.

\begin{definition}[Theories with prelude encoding]
We say that a theory $\T = \Sigma_{pre} \cup \Sigma_\T$ is a theory with prelude encoding when for every $c : A \in \Sigma_\T$, we have $\vdash A : \Type$.
\end{definition}
The condition guarantees that the user-defined constants of $\Sigma_{\T}$ are indeed encoded in the prelude encoding. For instance, we cannot define $\nat : \Type$, but are forced to take $\nat : \Set$. Consequently inside a theory with prelude encoding, the only constants $c : A \in \Sigma$ with $A$ a kind are $\Set$ (of type $\Type$), $\El$ (of type $\Set \ra \Type$) and $\Prf$ (of type $\El ~\o \ra \Type$). 

For each rewrite rule $\ell \lra r \in \Sigma$, the head-symbol of $\ell$ is a constant. It follows that, if $\Gamma \vdash \ell : A$, then $A$ cannot be $\Kind$. We thus have $\Gamma \vdash A : s$ with $s = \Type$ or $s = \Kind$. In particular, $\Type$ cannot occur in $\ell$ and $r$.

\begin{example}[Natural numbers]
We define a theory with prelude encoding $\T_n = \Sigma_{pre} \cup \Sigma_n$ for natural numbers. $\nat$ is the sort of natural numbers. We declare two constructors $0_n$ and $\succ_n$, a relation $\geq_n$, and an induction principle $\rec_n$.
\begin{flalign*}
&\begin{array}{ll}
\nat : &\Set \\
0_n : &\El ~\nat \\
\succ_n : &\El ~\nat \ra \El ~\nat \\
\geq_n ~: &\El ~\nat \ra \El ~\nat \ra \El ~\o \\
\ax^1_n : &\Pi (x : \El ~\nat). ~\Prf ~(x \geq_n x) \\
\ax^2_n : &\Pi (x : \El ~\nat). ~\Prf ~(\succ_n ~x \geq_n x) \\
\ax^3_n : &\Pi (x,y,z : \El ~\nat). ~\Prf ~(x \geq_n y) \ra \Prf ~(y \geq_n z) \ra \Prf ~(x \geq_n z) \\
\rec_n : &\Pi (P : \El ~\nat \ra \El ~\o). ~\Prf ~(P ~0_n) \ra \\
 &[\Pi (x : \El ~\nat). ~\Prf ~(P ~x) \ra \Prf ~(P ~(\succ_n ~x))] \ra \\
 &\Pi (x : \El ~\nat). ~\Prf ~(P ~x) 
\end{array} &
\end{flalign*}
In this theory, we can prove $\Pi (x : \El ~\nat). ~\Prf ~(x \geq_n 0_n)$ and $\Pi (x : \El ~\nat). ~\Prf ~(\succ_n ~x \geq_n 0_n)$.
\end{example}

\begin{example}[Integers]
We define a theory with prelude encoding $\T_i = \Sigma_{pre} \cup \Sigma_i$ for integers. $\int$ is the sort of integers. We declare three constructors $0_i$, $\succ_i$ and $\pred_i$, a relation $\geq_i$ and a generalized induction principle $\rec_i$.
\begin{flalign*}
&\begin{array}{ll}
\int : &\Set \\
0_i : &\El ~\int \\
\succ_i : &\El ~\int \ra \El ~\int \\
\pred_i : &\El ~\int \ra \El ~\int \\
\geq_i ~: &\El ~\int \ra \El ~\int \ra \El ~\o \\
\ax^1_i : &\Pi (x : \El ~\int). ~\Prf ~(x \geq_i x) \\
\ax^2_i : &\Pi (x : \El ~\int). ~\Prf ~(\succ_i ~x \geq_i x) \\
\ax^3_i : &\Pi (x,y,z : \El ~\int). ~\Prf ~(x \geq_i y) \ra \Prf ~(y \geq_i z) \ra \Prf ~(x \geq_i z) \\
\ax^4_i : &\Pi (x : \El ~\int). ~\Pi (P : \El ~\int \ra \El ~\o). ~\Prf ~(P ~(\succ_i ~(\pred_i ~x))) \ra \Prf ~(P ~x) \\
\ax^5_i : &\Pi (x : \El ~\int). ~\Pi (P : \El ~\int \ra \El ~\o). ~\Prf ~(P ~(\pred_i ~(\succ_i ~x))) \ra \Prf ~(P ~x) \\
\rec_i : &\Pi (c : \El ~\int)(P : \El ~\int \ra \El ~\o). ~\Prf ~(P ~c) \ra \\
&[\Pi (x : \El ~\int). ~\Prf ~(x \geq_i c) \ra \Prf ~(P ~x) \ra \Prf ~(P ~(\succ_i ~x))] \ra \\
&\Pi (x : \El ~\int). ~\Prf ~(x \geq_i c) \ra \Prf ~(P ~x) 
\end{array} &
\end{flalign*}
In this theory, we cannot prove $\Pi (x : \El ~\int). ~\Prf ~(x \geq_i 0_i)$ and $\Pi (x : \El ~\int). ~\Prf ~(\succ_i ~x \geq_i 0_i)$, but we can prove $\Pi (x : \El ~\int). ~\Prf ~(x \geq_i 0_i) \ra \Prf ~(\succ_i ~x \geq_i 0_i)$.
\end{example}

\section{Interpretation in the \texorpdfstring{$\lambda\Pi$}{lambdaPi}-Calculus Modulo Theory}
\label{sec_interp}

In this section, we define an interpretation of theories with prelude encoding. To do so, we first define the interpretation for the terms of the \lpcm, and then we extend it to theories with prelude encoding. Such an interpretation requires external parameters. In particular, we provide the parameters necessary for interpreting the prelude encoding. We show how the interpretation of a source theory $\sth$ in a target theory $\tth$ can be used to derive in $\tth$ the theorems proved in $\sth$. We conclude with an example: we provide the formal parameters for interpreting the theory of natural numbers $\T_n$ in the theory of integers $\T_i$.

\subsection{Interpretation of Terms}

\paragraph{Intuition.} When we interpret the source theory $\sth$ in the target theory $\tth$, we want to represent every term $t$ of $\sth$ by a term $\ts{t}$ in $\tth$, such that if $t$ is of type $A$ in $\sth$ then $\ts{t}$ is of type $\ts{A}$ in $\tth$. For instance, when interpreting the theory of natural numbers $\T_n$ in the theory of integers $\T_i$, we have to represent $\El ~\nat$ by $\ts{(\El ~\nat)}$. We would like to take $\ts{(\El ~\nat)} \coloneqq \Sigma (z : \El ~\int). ~\Prf ~(z \geq_i 0_i)$. However, the \lpcm does not feature $\Sigma$-types, and it is therefore difficult to express $\ts{(\El ~\nat)}$ in $\T_i$.

An alternative is to interpret the type of natural numbers $\El ~\nat$ by the type of integers $\El ~\int$, but we must guarantee that every integer representing a natural number is indeed non-negative. We naturally interpret the sort $\nat$ by $\int$, $0_n$ by $0_i$, $\succ_n$ by $\succ_i$, and $\geq_n$ by $\geq_i$. The interpretation of the theorem $\Pi (x : \El ~\nat). ~\Prf ~(\succ_n ~x \geq_n 0_n)$ should not be $\Pi (\ts{x} : \El ~\int). ~\Prf ~(\succ_i ~\ts{x} \geq_i 0_i)$, which is generally false for integers. Instead, we must ensure that $\ts{x}$ is an integer corresponding to a natural number, meaning that we suppose a proof of $\Prf ~(\ts{x} \geq_i 0_i)$. Thus the interpretation of the theorem $\Pi (x : \El ~\nat). ~\Prf ~(\succ_n ~x \geq_n 0_n)$ should be $\Pi (\ts{x} : \El ~\int). ~\Prf ~(\ts{x} \geq_i 0_i) \ra \Prf ~(\succ_i ~\ts{x} \geq_i 0_i)$.

\paragraph{Formal definition.} Following this intuition, when interpreting a term $t$ of type $A$ in $\sth$ by a term $\ts{t}$ of type $\ts{A}$ in $\tth$, we must take into account that $\ts{A}$ is a type that encompasses $A$, but may be larger than $A$. In that respect, we introduce another term $\tp{t}$ of type $\tp{A} ~\ts{t}$, where $\tp{A}$ is a predicate asserting that a given object of type $\ts{A}$ satisfies the semantic of type $A$. 

The interpretation of every constant $c$ is given by two parameters $\ts{c}$ and $\tp{c}$. The translation of an application $\ts{(t ~u)}$ is $\ts{t} ~\ts{u} ~\tp{u}$, since $\ts{t}$ takes as arguments $\ts{u}$ but also the witness $\tp{u}$. Similarly, $\tp{(t ~u)}$ is given by $\tp{t} ~\ts{u} ~\tp{u}$. If the variable $x$ occurs in $t$, then $\ts{x}$ and $\tp{x}$ may occur in $\ts{t}$ and $\tp{t}$. Hence $\ts{(\lambda (x : A). ~t)}$ is given by $\lambda (\ts{x} : \ts{A})(\tp{x} : \tp{A} ~\ts{x}). ~\ts{t}$ and $\tp{(\lambda (x : A). ~t)}$ is given by $\lambda (\ts{x} : \ts{A})(\tp{x} : \tp{A} ~\ts{x}). ~\tp{t}$. 

The same intuition holds for dependent types $\ts{(\Pi (x : A). ~B)}$. The predicate $\tp{(\Pi (x : A). ~B)}$ asserts that an object $f$ of type $\ts{(\Pi (x : A). ~B)}$ corresponds to the semantic of $\Pi (x : A). ~B$. In other words, for every $\ts{x}$ of type $\ts{A}$ and $\tp{x}$ of type $\tp{A} ~\ts{x}$, the term $f ~\ts{x} ~\tp{x}$ should satisfy the predicate $\tp{B}$. When $B$ is of type $\Type$, we take $\tp{(\Pi (x : A). ~B)} \coloneqq \lambda (f : \ts{(\Pi (x : A). ~B)}). ~\Pi (\ts{x} : \ts{A})(\tp{x} : \tp{A} ~\ts{x}). ~\tp{B} ~(f ~\ts{x} ~\tp{x})$. However, we cannot do the same when $B$ is of type $\Kind$, because this term would be ill-typed. Indeed, $\ts{(\Pi (x : A). ~B)}$ has type $\Kind$, while the type of the bound variable $f$ must have type $\Type$. To get around this issue, we introduce metavariables. We write $T\{ X \}$ when the metavariable $X$ occurs in $T$, and we write $T\{ t \}$ for the term obtained when substituting $X$ by $t$ in $T$. When $B$ has type $\Kind$, we take $\tp{(\Pi (x : A). ~B)}\{X\} \coloneqq \Pi (\ts{x} : \ts{A})(\tp{x} : \tp{A} ~\ts{x}). ~\tp{B}\{X ~\ts{x} ~\tp{x}\}$. Metavariables are only used for this purpose. In particular, they are always substituted and they never appear in typed terms.

\begin{definition}[Interpretation of terms]
The interpretation of terms of the \lpcm is given by the function $t \mapsto \ts{t}$ defined inductively by
\begin{flalign*}
&\begin{array}{l}
	\ts{(x)} \coloneqq \ts{x} \text{ (variable)} \\
	\ts{(c)} \coloneqq \ts{c} \text{ (parameter)} \\
	\ts{\Type} \coloneqq \Type \\
	\ts{\Kind} \coloneqq \Kind \\
	\ts{(t ~u)} \coloneqq \ts{t} ~\ts{u} ~\tp{u} \\
	\ts{(\lambda (x : A). ~t)} \coloneqq \lambda (\ts{x} : \ts{A})(\tp{x} : \tp{A} ~\ts{x}). ~\ts{t} \\
	\ts{(\Pi (x : A). ~B)} \coloneqq \Pi (\ts{x} : \ts{A})(\tp{x} : \tp{A} ~\ts{x}). ~\ts{B} 
\end{array} &
\end{flalign*}
and by the function $t \mapsto \tp{t}$ defined inductively by
\begin{flalign*}
&\begin{array}{l}
	\tp{(x)} \coloneqq \tp{x} \text{ (variable)} \\
	\tp{(c)} \coloneqq \tp{c} \text{ (parameter)} \\
	\tp{\Type}\{X\} \coloneqq X \ra \Type \\
	\tp{\Kind} \coloneqq \Kind \\
	\tp{(t ~u)} \coloneqq \tp{t} ~\ts{u} ~\tp{u} \\
	\tp{(\lambda (x : A). ~t)} \coloneqq \lambda (\ts{x} : \ts{A})(\tp{x} : \tp{A} ~\ts{x}). ~\tp{t} \\
	\tp{(\Pi (x : A). ~B)} \coloneqq \lambda (f : \ts{(\Pi (x : A). ~B)}). ~\Pi (\ts{x} : \ts{A})(\tp{x} : \tp{A} ~\ts{x}). ~\tp{B} ~(f ~\ts{x} ~\tp{x}) \text{ if $B : \Type$}  \\
    \tp{(\Pi (x : A). ~B)}\{X\} \coloneqq \Pi (\ts{x} : \ts{A})(\tp{x} : \tp{A} ~\ts{x}). ~\tp{B}\{X ~\ts{x} ~\tp{x}\} \text{ if $B : \Kind$.}
\end{array} &
\end{flalign*}
where the $X$ is a metavariable. The interpretation is extended to contexts with 
\begin{flalign*}
&\begin{array}{l}
	\tsp{\langle \rangle} \coloneqq \langle \rangle \\
	\tsp{(\Gamma, x : A)} \coloneqq \tsp{\Gamma}, \ts{x} : \ts{A}, \tp{x} : \tp{A} ~\ts{x}.
\end{array} &
\end{flalign*}
\end{definition}

When the free variable $x$ occurs in $t$, then $\ts{x}$ and $\tp{x}$ may both occur in $\ts{t}$ and $\tp{t}$. As such, we do not define distinct translations $\ts{\Gamma}$ and $\tp{\Gamma}$, but a single translation $\tsp{\Gamma}$, such that if $(x : A) \in \Gamma$ then $(\ts{x} : \ts{A}) \in \tsp{\Gamma}$ and $(\tp{x} : \tp{A} ~\ts{x}) \in \tsp{\Gamma}$.

\paragraph{Parametricity.} Remark that our interpretation is intuitively related to the parametricity translation~\cite{parametricity_conf}. Using parametricity, the translation $\ts{(t ~u)}$ is given by $\ts{t} ~\ts{u}$, the translation $\ts{(\lambda (x : A). ~t)}$ is given by $\lambda (\ts{x} : \ts{A}). ~\ts{t}$, and the translation $\ts{(\Pi (x : A). ~B)}$ is given by $\Pi (\ts{x} : \ts{A}). ~\ts{B}$. In our interpretation, we focus on embeddings and we want to represent every type $A$ of the source theory by a type $\ts{A}$ of the target theory. While $\Sigma$-types are well-suited for expressing such $\ts{A}$, they are not defined in the \lpcm. That is why we have applied a \textit{currying} operation on $\Sigma$-types. We therefore represent type $A$ using a more general type $\ts{A}$, and we guarantee that each term of type $\ts{A}$ representing a term of type $A$ enjoys the predicate $\tp{A}$. Consequently, the translation $\ts{(\Pi (x : A). ~B)}$ is given by $\Pi (\ts{x} : \ts{A})(\tp{x} : \tp{A} ~\ts{x}). ~\ts{B}$, the translation $\ts{(\lambda (x : A). ~t)}$ is given by $\lambda (\ts{x} : \ts{A})(\tp{x} : \tp{A} ~\ts{x}). ~\ts{t}$, and the translation $\ts{(t ~u)}$ is given by $\ts{t} ~\ts{u} ~\tp{u}$. The formal relation between the parametricity translation and our interpretation remains to be investigated.

\subsection{Parameters for the Prelude Encoding}

We aim at interpreting a source theory $\sth$ in a target theory $\tth$, when $\sth$ and $\tth$ are theories with prelude encoding. Such an interpretation is parametrized by the terms of $\tth$ that correspond to the constants of $\sth$. In particular, we have to provide the parameters for the constants of the prelude encoding.

When $\vdash t : A$ in $\sth$, we want to have $\vdash \ts{t} : \ts{A}$ in $\tth$. Moreover, we want $\vdash \tp{A} : \ts{A} \ra \Type$ in $\tth$ when $A = \Type$. These conditions lead to the definition of $\ts{\Set}$, $\tp{\Set}$, $\ts{\El}$, $\tp{\El}$, $\ts{\Prf}$, $\tp{\Prf}$ and $\ts{\o}$. When $t$ is of type $\Prf ~p$, we need a witness $\tp{t}$ of type $\tp{(\Prf ~p)} ~\ts{t}$ asserting that $\ts{t}$ is indeed a proof of $\ts{p}$. Since $\ts{t}$ is of type $\Prf ~\ts{p}$, it is necessarily a proof of $\ts{p}$, and we define $\tp{\Prf}$ so that we can always choose $\tp{t}$ to be $\ts{t}$. The predicate $\tp{\o}$ asserts that an object $\ts{p}$ of type $\El ~\o$ is indeed a proposition, so we choose $\tp{\o}$ to be $\lambda (z : \El ~\o). ~z \impd (\lambda (x : \Prf ~z). ~z)$. Consequently, it is is always possible to find a witness $\tp{p}$ of type $\Prf ~(\tp{\o} ~\ts{p})$, that is $\Prf ~\ts{p} \ra \Prf ~\ts{p}$.
\begin{flalign*}
&\begin{array}{l}
\ts{\Set} \coloneqq \Set \\
\tp{\Set} \coloneqq \lambda (x : \Set). ~\El ~x \ra \El ~\o \\
\ts{\o} \coloneqq o \\
\tp{\o} \coloneqq \lambda (z : \El ~\o). ~z \impd (\lambda (x : \Prf ~z). ~z) \\
\ts{\El} \coloneqq \lambda (\ts{x} : \Set)(\tp{x} : \El ~\ts{x} \ra \El ~\o). ~\El ~\ts{x} \\
\tp{\El} \coloneqq \lambda (\ts{u} : \Set)(\tp{u} : \El ~\ts{u} \ra \El ~\o)(x : \El ~\ts{u}). ~\Prf ~(\tp{u} ~x) \\
\ts{\Prf} \coloneqq \lambda (\ts{x} : \El ~\o)(\tp{x} : \Prf ~(\tp{\o} ~\ts{x})). ~\Prf ~\ts{x} \\
\tp{\Prf} \coloneqq \lambda (\ts{u} : \El~\o)(\tp{u} : \Prf ~(\tp{\o} ~\ts{u}))(x : \Prf ~\ts{u}). ~\Prf ~\ts{u} 
\end{array} &
\end{flalign*}
Parameters $\ts{\arrd}$ and $\tp{\arrd}$ are defined so that $(\El ~(a \arrd b))^@ \equiv_{\beta\Sigma} (\Pi (x : \El ~a). ~\El ~(b ~x))^@$ for $@ \in \{ *, + \}$.
\begin{flalign*}
&\begin{array}{ll}
\ts{\arrd} \coloneqq &\lambda (\ts{a} : \Set)(\tp{a} : \El ~\ts{a} \ra \El ~\o)(\ts{b} : \Pi (\ts{x} : \El ~\ts{a}). ~\Prf ~(\tp{a} ~\ts{x}) \ra \Set). \\
&\lambda (\tp{b} : \Pi (\ts{x} : \El ~\ts{a})(\tp{x} : \Prf ~(\tp{a} ~\ts{x})). ~\El ~(\ts{b} ~\ts{x} ~\tp{x}) \ra \El ~\o). \\
&\ts{a} \arrd (\lambda (\ts{x} : \El ~\ts{a}). ~\blpi ~(\tp{a} ~\ts{x}) ~(\ts{b} ~\ts{x})) \\
\end{array} &\\
&\begin{array}{ll}
\tp{\arrd} \coloneqq &\lambda (\ts{a} : \Set)(\tp{a} : \El ~\ts{a} \ra \El ~\o)(\ts{b} : \Pi (\ts{x} : \El ~\ts{a}). ~\Prf ~(\tp{a} ~\ts{x}) \ra \Set). \\
&\lambda (\tp{b} : \Pi (\ts{x} : \El ~\ts{a})(\tp{x} : \Prf ~(\tp{a} ~\ts{x})). ~\El ~(\ts{b} ~\ts{x} ~\tp{x}) \ra \El ~\o). \\
&\lambda (f : \El ~\ts{(a \arrd b)}). \\
&\fa ~\ts{a} ~(\lambda (\ts{x} : \El ~\ts{a}). ~(\tp{a} ~\ts{x}) \impd (\lambda (\tp{x} : \Prf ~(\tp{a} ~\ts{x})). ~\tp{b} ~\ts{x} ~\tp{x} ~(f ~\ts{x} ~\tp{x}))) \\
\end{array} &
\end{flalign*}
Parameter $\ts{\impd}$ is defined so that $\ts{(\Prf ~(a \impd b))} \equiv_{\beta\Sigma} \ts{(\Pi (x : \Prf ~a). ~\Prf ~(b ~x))}$. Because the condition $\tp{(\Prf ~(a \impd b))} \equiv_{\beta\Sigma} \tp{(\Pi (x : \Prf ~a). ~\Prf ~(b ~x))}$ holds regardless of the definition of $\tp{\impd}$, we choose $\tp{\impd}$ so that $\vdash \tp{\impd} : \tp{(\Pi (a : \El ~\o). ~(\Prf ~a \ra \El ~\o) \ra \El ~\o)} ~\ts{\impd}$. 
\begin{flalign*}
&\begin{array}{ll}
\ts{\impd} \coloneqq &\lambda (\ts{a} : \El ~\o)(\tp{a} : \Prf ~(\tp{\o} ~\ts{a}))(\ts{b} : \Pi (\ts{x} : \Prf ~\ts{a}). ~\Prf ~\ts{a} \ra \El ~\o). \\
&\lambda (\tp{b} : \Pi (\ts{x} : \Prf ~\ts{a})(\tp{x} : \Prf ~\ts{a}). ~\Prf ~(\tp{\o} ~(\ts{b} ~\ts{x} ~\tp{x}))). \\
&\ts{a} \impd (\lambda (\ts{x} : \Prf ~\ts{a}). ~\ts{a} \impd ~(\ts{b} ~\ts{x})) \\
\end{array} \\
&\begin{array}{ll}
\tp{\impd} \coloneqq &\lambda (\ts{a} : \El ~\o)(\tp{a} : \Prf ~(\tp{\o} ~\ts{a}))(\ts{b} : \Pi (\ts{x} : \Prf ~\ts{a}). ~\Prf ~\ts{a} \ra \El ~\o). \\
&\lambda (\tp{b} : \Pi (\ts{x} : \Prf ~\ts{a})(\tp{x} : \Prf ~\ts{a}). ~\Prf ~(\tp{\o} ~(\ts{b} ~\ts{x} ~\tp{x}))). \\
&\lambda (p : \Prf ~\ts{(a \impd b)}). ~p \\
\end{array} &
\end{flalign*}
Parameters $\ts{\blpi}$ and $\tp{\blpi}$ are defined so that $(\El ~(\blpi ~a ~b))^@ \equiv_{\beta\Sigma} (\Pi (x : \Prf ~a). ~\El ~(b ~x))^@$ for $@ \in \{ *, + \}$. 
\begin{flalign*}
&\begin{array}{ll}
\ts{\blpi} \coloneqq &\lambda (\ts{a} : \El ~\o)(\tp{a} : \Prf ~(\tp{\o} ~\ts{a}))(\ts{b} : \Pi (\ts{x} : \Prf ~\ts{a}). ~\Prf ~\ts{a} \ra \Set). \\
&\lambda (\tp{b} : \Pi (\ts{x} : \Prf ~\ts{a})(\tp{x} : \Prf ~\ts{a}). ~\El ~(\ts{b} ~\ts{x} ~\tp{x}) \ra \El ~\o). \\
&\blpi ~\ts{a} ~(\lambda (\ts{x} : \Prf ~\ts{a}). ~\blpi ~\ts{a} ~(\ts{b} ~\ts{x})) \\
\end{array} \\
&\begin{array}{ll}
\tp{\blpi} \coloneqq &\lambda (\ts{a} : \El ~\o)(\tp{a} : \Prf ~(\tp{\o} ~\ts{a}))(\ts{b} : \Pi (\ts{x} : \Prf ~\ts{a}). ~\Prf ~\ts{a} \ra \Set). \\
&\lambda (\tp{b} : \Pi (\ts{x} : \Prf ~\ts{a})(\tp{x} : \Prf ~\ts{a}). ~\El ~(\ts{b} ~\ts{x} ~\tp{x}) \ra \El ~\o). \\
&\lambda (f : \El ~\ts{(\blpi ~a ~b)}). \\
&\ts{a} \impd ~(\lambda (\ts{x} : \Prf ~\ts{a}). ~\ts{a} \impd (\lambda (\tp{x} : \Prf ~\ts{a}). ~\tp{b} ~\ts{x} ~\tp{x} ~(f ~\ts{x} ~\tp{x}))) \\
\end{array} &
\end{flalign*}
Parameter $\ts{\fa}$ is defined so that $\ts{(\Prf ~(\fa ~a ~b))} \equiv_{\beta\Sigma} \ts{(\Pi (x : \El ~a). ~\Prf ~(b ~x))}$. Because the condition $\tp{(\Prf ~(\fa ~a ~b))} \equiv_{\beta\Sigma} \tp{(\Pi (x : \El ~a). ~\Prf ~(b ~x))}$ holds regardless of the definition of $\tp{\fa}$, we choose $\tp{\fa}$ so that $\vdash \tp{\fa} : \tp{(\Pi (a : \Set). ~(\El ~a \ra \El ~\o) \ra \El ~\o)} ~\ts{\fa}$.
\begin{flalign*}
&\begin{array}{ll}
\ts{\fa} \coloneqq &\lambda (\ts{a} : \Set)(\tp{a} : \El ~\ts{a} \ra \El ~\o)(\ts{b} : \Pi (\ts{x} : \El ~\ts{a}). ~\Prf ~(\tp{a} ~\ts{x}) \ra \El ~\o). \\
&\lambda (\tp{b} : \Pi (\ts{x} : \El ~\ts{a})(\tp{x} : \Prf ~(\tp{a} ~\ts{x})). ~\Prf ~(\tp{\o} ~(\ts{b} ~\ts{x} ~\tp{x}))). \\
&\fa ~\ts{a} ~(\lambda (\ts{x} : \El ~\ts{a}). ~(\tp{a} ~\ts{x}) \impd (\ts{b} ~\ts{x})) \\
\end{array} \\
&\begin{array}{ll}
\tp{\fa} \coloneqq &\lambda (\ts{a} : \Set)(\tp{a} : \El ~\ts{a} \ra \El ~\o)(\ts{b} : \Pi (\ts{x} : \El ~\ts{a}). ~\Prf ~(\tp{a} ~\ts{x}) \ra \El ~\o). \\
&\lambda (\tp{b} : \Pi (\ts{x} : \El ~\ts{a})(\tp{x} : \Prf ~(\tp{a} ~\ts{x})). ~\Prf ~(\tp{\o} ~(\ts{b} ~\ts{x} ~\tp{x}))). \\
&\lambda (p : \Prf ~\ts{(\fa ~a ~b)}). ~p \\
\end{array} &
\end{flalign*}
The parameters chosen for the constants of the prelude encoding satisfy the expected properties. For any $c : A \in \Sigma_{pre}$, we have $\vdash \ts{c} : \ts{A}$ and $\vdash \tp{c} : \tp{A} ~\ts{c}$. Moreover, the interpretation respects the conversion relation, meaning that for each rewrite rule $\ell \lra r$ of $\Sigma_{pre}$, we have both $\ts{\ell} \equiv_{\beta\Sigma} \ts{r}$ and $\tp{\ell} \equiv_{\beta\Sigma} \tp{r}$.

\begin{proposition}
Let $c : A \in \Sigma_{pre}$.
\begin{enumerate}
\item We have $\vdash \ts{c} : \ts{A}$.
\item 
\begin{enumerate}
\item If $\vdash A : \Type$ then $\vdash \tp{c} : \tp{A} ~\ts{c}$.
\item If $\vdash A : \Kind$ then $\vdash \tp{c} : \tp{A}\{\ts{c}\}$.
\end{enumerate}
\end{enumerate}
\end{proposition}

\begin{proof}
By simple verification. The result has been checked in \dk, see the definitions of the parameters in the file \texttt{lo\_sp.dk}\footnote{All the \dk files are available at \url{https://github.com/thomastraversie/InterpDK}.}.
\end{proof}

\begin{proposition}
\label{prop_rules}
For every $\ell \lra r \in \Sigma_{pre}$, we have $\ts{\ell} \equiv_{\beta\Sigma} \ts{r}$ and $\tp{\ell} \equiv_{\beta\Sigma} \tp{r}$.
\end{proposition}

\begin{proof}
We only show the case $\El ~(a \arrd b) \lra \Pi (x : \El ~a). ~\El ~(b ~x)$.
\begin{flalign*}
&\begin{array}{lll}
	\text{We have } \ts{(\El ~(a \arrd b))} &\equiv_{\beta\Sigma} &\El ~\ts{(a \arrd b)} \\
	 &\equiv_{\beta\Sigma} &\El ~(\ts{a} \arrd (\lambda \ts{x}. ~\blpi ~(\tp{a} ~\ts{x}) ~(\ts{b} ~\ts{x}))) \\
	 &\equiv_{\beta\Sigma} &\Pi (\ts{x} : \El ~\ts{a}). ~\El ~(\blpi ~(\tp{a} ~\ts{x}) ~(\ts{b} ~\ts{x})) \\
	 &\equiv_{\beta\Sigma} &\Pi (\ts{x} : \El ~\ts{a})(\tp{x} : \Prf ~(\tp{a} ~\ts{x})). ~\El ~(\ts{b} ~\ts{x} ~\tp{x}) \\
	 &\equiv_{\beta\Sigma} &\Pi (\ts{x} : \ts{(\El ~a)})(\tp{x} : \tp{(\El ~a)} ~ \ts{x}). ~\ts{(\El ~(b ~x))} \\
	 &\equiv_{\beta\Sigma} &\ts{(\Pi (x : \El ~a). ~\El ~(b ~x))}
\end{array} & \\
&\begin{array}{lll}
	\text{and } \tp{(\El ~(a \arrd b))} &\equiv_{\beta\Sigma} &\lambda (f : \El ~\ts{(a \arrd b)}). ~\Prf ~(\tp{(a \arrd b)} ~f) \\
	 &\equiv_{\beta\Sigma} &\lambda (f : \El ~\ts{(a \arrd b)}). \\
	 & &\Prf ~(\fa ~\ts{a} ~(\lambda \ts{x}. ~(\tp{a} ~\ts{x}) \impd (\lambda \tp{x}. ~\tp{b} ~\ts{x} ~\tp{x} ~(f ~\ts{x} ~\tp{x})))) \\
	 &\equiv_{\beta\Sigma} &\lambda (f : \El ~\ts{(a \arrd b)}). ~\Pi (\ts{x} : \El ~\ts{a}). \\
	 & &\Prf ~((\tp{a} ~\ts{x}) \impd (\lambda \tp{x}. ~\tp{b} ~\ts{x} ~\tp{x} ~(f ~\ts{x} ~\tp{x}))) \\
	 &\equiv_{\beta\Sigma} &\lambda (f : \ts{(\El ~(a \arrd b))}). ~\Pi (\ts{x} : \El ~\ts{a})(\tp{x} : \Prf ~(\tp{a} ~\ts{x})). \\
	 & &\Prf ~(\tp{b} ~\ts{x} ~\tp{x} ~(f ~\ts{x} ~\tp{x})) \\
	 &\equiv_{\beta\Sigma} &\lambda (f : \ts{(\Pi (x : \El ~a). ~\El ~(b ~x))}). ~\Pi (\ts{x} : \ts{(\El ~a)})(\tp{x} : \tp{(\El ~a)} ~ \ts{x}). \\
	 & &\tp{(\El ~(b ~x))} ~(f ~\ts{x} ~\tp{x}) \\
	 &\equiv_{\beta\Sigma} &\tp{(\Pi (x : \El ~a). ~\El ~(b ~x))}.
\end{array} &
\end{flalign*}
The result has been checked in \dk for the four rewrite rules, see the \texttt{\#ASSERT} commands in the file \texttt{lo\_sp.dk}.
\end{proof}

\subsection{Interpretation of Theories}

The interpretation of a source theory $\sth$ in a target theory $\tth$ is given by the parameters $\ts{c}$ and $\tp{c}$, for each constant $c$ of $\Sigma$. We have provided the parameters for the constants of $\Sigma_{pre}$, but the parameters for the constants of $\Sigma_{\sth}$ remain to be given by the user.

\begin{definition}[Interpretation of theories]
\label{def_interp}
Let $\sth$ and $\tth$ be two theories with prelude encoding. We say that $\sth$ has an interpretation in $\tth$ when:
\begin{enumerate}
\item for each constant $c : A \in \Sigma_{\sth}$, we have a term $\ts{c}$ such that $\vdash \ts{c} : \ts{A}$ in $\tth$,

\item for each constant $c : A \in \Sigma_{\sth}$, we have a term $\tp{c}$ such that $\vdash \tp{c} : \tp{A} ~\ts{c}$ in $\tth$,

\item for each rewrite rule $\ell \lra r \in \Sigma_{\sth}$, we have $\ts{\ell} \equiv_{\beta\Sigma} \ts{r}$ and $\tp{\ell} \equiv_{\beta\Sigma} \tp{r}$ in $\tth$.
\end{enumerate}
\end{definition}

Remark that, in the third item, $\tp{\ell}$ and $\tp{r}$ do not contain metavariables, as we have seen that $\Type$ cannot occur in $\ell$ and $r$.

If we cannot interpret the rewrite rules of $\sth$ into conversions in $\tth$, we can nonetheless replace the rewrite rules of $\sth$ by equational axioms---that is by typed constants---and then interpret such constants in $\tth$. So as to replace user-defined rewrite rules by equational axioms~\cite{elimrule}, we add an equality in our signature, and we use functional extensionality, uniqueness of identity proofs, and the congruence of equality on applications.

The \lpcm features substitutions in the type of an application---in the case of dependent types---and features user-defined rewrite rules. So that the translation of a provable judgment remains provable, it is important to maintain substitution and conversion through the translations $t \mapsto \ts{t}$ and $t \mapsto \tp{t}$. For each variable $z$ occurring in a term $t$, the two variables $\ts{z}$ and $\tp{z}$ occur in the translated terms $\ts{t}$ and $\tp{t}$. The translation $\ts{(t[z \leftarrow w])}$ is thus given by $\ts{t}[\ts{z} \leftarrow \ts{w}][\tp{z} \leftarrow \tp{w}]$.

\begin{proposition}[Substitution]
\label{prop_subst}
Let $t$ and $w$ be two terms and $z$ be a variable. We have:
\begin{itemize}
\item $\ts{(t[z \leftarrow w])} = \ts{t}[\ts{z} \leftarrow \ts{w}][\tp{z} \leftarrow \tp{w}]$.
\item $\tp{(t[z \leftarrow w])} = \tp{t}[\ts{z} \leftarrow \ts{w}][\tp{z} \leftarrow \tp{w}]$.
\end{itemize}
\end{proposition}

\begin{proof}
By induction on the term $t$.
\end{proof}

\begin{proposition}[Conversion]
\label{prop_conv}
If $A \equiv_{\beta\Sigma} B$ in $\sth$, then $\ts{A} \equiv_{\beta\Sigma} \ts{B}$ and $\tp{A} \equiv_{\beta\Sigma} \tp{B}$ in $\tth$.
\end{proposition}

\begin{proof}
We prove the result by induction on the formation of $A \equiv_{\beta\Sigma} B$.
\begin{itemize}
\item We have $\ts{(\lambda (x : A). ~t) ~ u)} = (\lambda (\ts{x} : \ts{A})(\tp{x} : \tp{A} ~ \ts{x}). ~\ts{t}) ~\ts{u} ~\tp{u}$, which $\beta$-reduces to $\ts{t}[\ts{x} \leftarrow \ts{u}][\tp{x} \leftarrow \tp{u}]$, that is $\ts{(t[x \leftarrow u])}$ following \Cref{prop_subst}. Similarly, $\tp{((\lambda (x : A). ~t) ~ u)} \equiv_{\beta\Sigma} \tp{(t[x \leftarrow u])}$.

\item For each $\ell \lra r \in \Sigma$ and any substitution $\theta$, we have $\ts{\ell} \equiv_{\beta\Sigma} \ts{r}$ by definition and \Cref{prop_rules}. Using \Cref{prop_subst}, we have $\ts{(\ell\theta)} = \ts{\ell}\tsp{\theta}$ and $\ts{(r\theta)} = \ts{r}\tsp{\theta}$, where $\tsp{\theta}$ is defined so that if $\theta$ substitutes $z$ by $w$, then $\tsp{\theta}$ substitutes $\ts{z}$ by $\ts{w}$ and $\tp{z}$ by $\tp{w}$. Therefore $\ts{(\ell\theta)} = \ts{\ell}\tsp{\theta} \equiv_{\beta\Sigma} \ts{r}\tsp{\theta} = \ts{(r\theta)}$. Similarly, we have $\tp{(\ell\theta)} = \tp{\ell}\tsp{\theta} \equiv_{\beta\Sigma} \tp{r}\tsp{\theta} = \tp{(r\theta)}$. 

\item For closure by context, we only show the $\lambda$-abstraction case. Suppose that $\lambda (x : A). ~t \equiv_{\beta\Sigma} \lambda (x : B). ~u$ derives from $A \equiv_{\beta\Sigma} B$ and $t \equiv_{\beta\Sigma} u$. By induction, we have $\ts{A} \equiv_{\beta\Sigma} \ts{B}$, and $\tp{A} \equiv_{\beta\Sigma} \tp{B}$, and $\ts{t} \equiv_{\beta\Sigma} \ts{u}$, and $\tp{t} \equiv_{\beta\Sigma} \tp{u}$. We derive that $\lambda (\ts{x} : \ts{A})(\tp{x} : \tp{A} ~\ts{x}). ~\ts{t} \equiv_{\beta\Sigma} \lambda (\ts{x} : \ts{B})(\tp{x} : \tp{B} ~ \ts{x}). ~\ts{u}$, that is $\ts{(\lambda (x : A). ~t)} \equiv_{\beta\Sigma} \ts{(\lambda (x : B). ~u)}$. Similarly, $\tp{(\lambda (x : A). ~t)} \equiv_{\beta\Sigma} \tp{(\lambda (x : B). ~u)}$.
\item Reflexivity, symmetry and transitivity are immediate.
\end{itemize}
\end{proof}

We have at hand all the tools allowing us to prove that, when $\sth$ has an interpretation in $\tth$, any provable judgment in $\sth$ is interpreted as a provable judgment in $\tth$. The first item of the theorem concerns well-formedness judgments. The second item concerns typing judgments with respect to the translation $t \mapsto \ts{t}$, and the third item concerns typing judgments with respect to the translation $t \mapsto \tp{t}$.

\begin{theorem}[Interpretation]
\label{thm_interp}
Let $\sth$ and $\tth$ be two theories with prelude encoding, such that $\sth$ has an interpretation in $\tth$.
\begin{enumerate}
\item If $\vdash \Gamma$ in $\sth$, then $\vdash \tsp{\Gamma}$ in $\tth$.

\item If $\Gamma \vdash t : A$ in $\sth$ then $\tsp{\Gamma} \vdash \ts{t} : \ts{A}$ in $\tth$.

\item 
\begin{enumerate}
\item If $\Gamma \vdash t : A$ and $\Gamma \vdash A : \Type$ in $\sth$, then $\tsp{\Gamma} \vdash \tp{t} : \tp{A} ~\ts{t}$ in $\tth$.

\item If $\Gamma \vdash t : A$ and $\Gamma \vdash A : \Kind$ in $\sth$, then $\tsp{\Gamma} \vdash \tp{t} : \tp{A}\{\ts{t}\}$ in $\tth$.

\item If $\Gamma \vdash A : \Kind$ in $\sth$, then for every $t$ such that $\tsp{\Gamma} \vdash t : \ts{A}$ in $\tth$, we have $\tsp{\Gamma} \vdash \tp{A}\{t\} : \Kind$.
\end{enumerate}
\end{enumerate}
\end{theorem}

\begin{proof}
We proceed by induction on the derivation. We only show the most interesting cases.
\begin{itemize}
\item \underline{\textsc{Const:}} By induction, we have $\vdash \tsp{\Gamma}$ and $\vdash \ts{A} : \ts{s}$. Since $c : A \in \Sigma$, we have $\vdash \ts{c} : \ts{A}$. We derive $\tsp{\Gamma} \vdash \ts{c} : \ts{A}$ by weakening. If $s = \Type$, then $\vdash \tp{c} : \tp{A} ~\ts{c}$ and we derive $\tsp{\Gamma} \vdash \tp{c} : \tp{A} ~\ts{c}$ by weakening. If $s = \Kind$, then $\vdash \tp{c} : \tp{A}\{\ts{c}\}$ and we derive $\tsp{\Gamma} \vdash \tp{c} : \tp{A}\{\ts{c}\}$ by weakening.

\item \underline{\textsc{Prod:}} By induction, we have $\tsp{\Gamma} \vdash \ts{A} : \Type$, and $\tsp{\Gamma} \vdash \tp{A} : \ts{A} \ra \Type$, and $\tsp{\Gamma}, \ts{x} : \ts{A}, \tp{x} : \tp{A} ~\ts{x} \vdash \ts{B} : \ts{s}$. Using \textsc{Prod}, we get $\tsp{\Gamma} \vdash \Pi (\ts{x} : \ts{A})(\tp{x} : \tp{A} ~\ts{x}). ~\ts{B} : \ts{s}$. 

Suppose that $s = \Type$. By induction, $\tsp{\Gamma}, \ts{x} : \ts{A}, \tp{x} : \tp{A} ~\ts{x} \vdash \tp{B} : \ts{B} \ra \Type$. By weakening, we have $\tsp{\Gamma}, f : \ts{(\Pi (x : A). ~B)}, \ts{x} : \ts{A}, \tp{x} : \tp{A} ~\ts{x} \vdash \tp{B} : \ts{B} \ra \Type$. Since $\tsp{\Gamma}, f : \ts{(\Pi (x : A). ~B)}, \ts{x} : \ts{A}, \tp{x} : \tp{A} ~\ts{x} \vdash \tp{B} ~(f ~\ts{x} ~\tp{x}) : \Type$, we derive $\tsp{\Gamma} \vdash \lambda (f : \ts{(\Pi (x : A). ~B)}). ~\Pi (\ts{x} : \ts{A})(\tp{x} : \tp{A} ~\ts{x}). ~\tp{B} ~(f ~\ts{x} ~\tp{x}) : \ts{(\Pi (x : A). ~B)} \ra \Type$, which corresponds to $\tsp{\Gamma} \vdash \tp{(\Pi (x : A). ~B)} : \tp{\Type}\{\ts{(\Pi (x : A). ~B)}\}$.

Suppose that $s = \Kind$ and that we have $\tsp{\Gamma} \vdash t : \ts{(\Pi (x : A). ~B)}$. Since $\tsp{\Gamma}, \ts{x} : \ts{A}, \tp{x} : \tp{A} ~\ts{x} \vdash t ~\ts{x} ~\tp{x} : \ts{B}$, by induction we get $\tsp{\Gamma}, \ts{x} : \ts{A}, \tp{x} : \tp{A} ~\ts{x} \vdash \tp{B}\{t ~\ts{x} ~\tp{x}\} : \Kind$. We derive $\tsp{\Gamma} \vdash \Pi (\ts{x} : \ts{A})(\tp{x} : \tp{A} ~\ts{x}). ~\tp{B}\{t ~\ts{x} ~\tp{x}\} : \Kind$, that is $\tsp{\Gamma} \vdash \tp{(\Pi (x : A). ~B)}\{t\} : \Kind$.

\item \underline{\textsc{Abs:}} By induction, we have $\tsp{\Gamma} \vdash \ts{A} : \Type$, and $\tsp{\Gamma} \vdash \tp{A} : \ts{A} \ra \Type$, and $\tsp{\Gamma}, \ts{x} : \ts{A}, \tp{x} : \tp{A} ~\ts{x} \vdash \ts{B} : \ts{s}$, and $\tsp{\Gamma}, \ts{x} : \ts{A}, \tp{x} : \tp{A} ~\ts{x} \vdash \ts{t} : \ts{B}$, . We derive $\tsp{\Gamma} \vdash \lambda (\ts{x} : \ts{A})(\tp{x} : \tp{A} ~\ts{x}). ~\ts{t} : \Pi (\ts{x} : \ts{A})(\tp{x} : \tp{A} ~\ts{x}). ~\ts{B}$, that is $\tsp{\Gamma} \vdash \ts{(\lambda (x : A). ~t)} : \ts{(\Pi (x : A). ~B)}$.

Suppose that $s = \Type$. By induction, we have $\tsp{\Gamma}, \ts{x} : \ts{A}, \tp{x} : \tp{A} ~\ts{x} \vdash \tp{B} : \ts{B} \ra \Type$ and $\tsp{\Gamma}, \ts{x} : \ts{A}, \tp{x} : \tp{A} ~\ts{x} \vdash \tp{t} : \tp{B} ~\ts{t}$. We derive $\tsp{\Gamma} \vdash \lambda (\ts{x} : \ts{A})(\tp{x} : \tp{A} ~\ts{x}). ~\tp{t} : \Pi (\ts{x} : \ts{A})(\tp{x} : \tp{A} ~\ts{x}). ~\tp{B} ~\ts{t}$. Using \textsc{Conv}, we conclude that $\tsp{\Gamma} \vdash \tp{(\lambda (x : A). ~t)} : \tp{(\Pi (x : A). ~B)} ~\ts{(\lambda (x : A). ~t)}$.

Suppose that $s = \Kind$. By induction, we have $\tsp{\Gamma}, \ts{x} : \ts{A}, \tp{x} : \tp{A} ~\ts{x} \vdash \tp{B}\{\ts{t}\} : \Kind$ and $\tsp{\Gamma}, \ts{x} : \ts{A}, \tp{x} : \tp{A} ~\ts{x} \vdash \tp{t} : \tp{B}\{\ts{t}\}$. We derive $\tsp{\Gamma} \vdash \lambda (\ts{x} : \ts{A})(\tp{x} : (\tp{A} ~\ts{x})). ~\tp{t} : \Pi (\ts{x} : \ts{A})(\tp{x} : \tp{A} ~\ts{x}). ~\tp{B}\{\ts{t}\}$, that is $\tsp{\Gamma} \vdash \tp{(\lambda (x : A). ~t)} : \tp{(\Pi (x : A). ~B)}\{\ts{(\lambda (x : A). ~t)}\}$ using \textsc{Conv}.

\item \underline{\textsc{App:}} By induction, we have $\tsp{\Gamma} \vdash \ts{t} : \Pi (\ts{x} : \ts{A})(\tp{x} : \tp{A} ~\ts{x}). ~\ts{B}$, and $\tsp{\Gamma} \vdash \ts{u} : \ts{A}$, and $\tsp{\Gamma} \vdash \tp{u} : \tp{A} ~\ts{u}$. We derive $\tsp{\Gamma} \vdash \ts{t} ~\ts{u} ~\tp{u} : \ts{B}[\ts{x} \leftarrow \ts{u}][\tp{x} \leftarrow \tp{u}]$. Using \Cref{prop_subst}, we conclude that $\tsp{\Gamma} \vdash \ts{(t ~u)} : \ts{(B[x \leftarrow u])}$.

Suppose that $\Gamma \vdash \Pi (x : A). ~B : \Type$ (and thus $\Gamma \vdash B : \Type$). By induction, we have $\tsp{\Gamma} \vdash \tp{t} : \Pi (\ts{x} : \ts{A})(\tp{x} : \tp{A} ~\ts{x}). ~\tp{B} ~(\ts{t} ~\ts{x} ~\tp{x})$. It follows that $\tsp{\Gamma} \vdash \tp{t} ~\ts{u} ~\tp{u} : \tp{B}[\ts{x} \leftarrow \ts{u}][\tp{x} \leftarrow \tp{u}] ~(\ts{t} ~\ts{u} ~\tp{u})$. Using \Cref{prop_subst}, we conclude that $\tsp{\Gamma} \vdash \tp{(t ~u)} : \tp{(B[x \leftarrow u])} ~\ts{(t ~u)}$.

Suppose that $\Gamma \vdash \Pi (x : A). ~B : \Kind$ (and thus $\Gamma \vdash B : \Kind$). By induction, we have $\tsp{\Gamma} \vdash \tp{t} : \Pi (\ts{x} : \ts{A})(\tp{x} : \tp{A} ~\ts{x}). ~\tp{B}\{\ts{t} ~\ts{x} ~\tp{x}\}$. It follows that $\tsp{\Gamma} \vdash \tp{t} ~\ts{u} ~\tp{u} : (\tp{B}\{\ts{t} ~\ts{x} ~\tp{x}\})[\ts{x} \leftarrow \ts{u}][\tp{x} \leftarrow \tp{u}]$. Using \Cref{prop_subst}, we conclude that $\tsp{\Gamma} \vdash \tp{(t ~u)} : \tp{(B[x \leftarrow u])}\{\ts{(t ~u)}\}$.

\item \underline{\textsc{Conv:}} We conclude using the induction hypotheses and \Cref{prop_conv}.
\end{itemize}
\end{proof}

Given an interpretation of a source theory $\sth$ in a target theory $\tth$, the results proved in $\sth$ are automatically transported to $\tth$. The interpretation of $\sth$ in $\tth$ only requires the parameters $\ts{c}$ and $\tp{c}$ in $\tth$ for each user-defined constant $c$ of $\sth$. Once we have an interpretation of $\sth$ in $\tth$, it is possible to prove that $\sth$ is consistent provided that $\tth$ is so. In the \lpcm, we say that a theory is inconsistent when we can build a term that takes a proposition and returns one of its proofs, that is when there exists a term $t$ such that $\vdash t : \Pi (P : \El ~\o). ~\Prf ~P$.

\begin{theorem}[Relative consistency]
Let $\sth$ and $\tth$ be two theories with prelude encoding, such that $\sth$ has an interpretation in $\tth$. If $\tth$ is consistent, then $\sth$ is consistent too.
\end{theorem}

\begin{proof}
Assume that $\sth$ is inconsistent, meaning that we have a term $\vdash t : \Pi (P : \El ~\o). ~\Prf ~P$. By applying \Cref{thm_interp}, we get $\vdash t : \Pi (\ts{P} : \El ~\o)(\tp{P} : \Prf ~\ts{P} \ra \Prf ~\ts{P}). ~\Prf ~\ts{P}$. We take the term $t' \coloneqq \lambda (\ts{P} : \El ~\o). ~t ~\ts{P} ~(\lambda (x : \Prf ~\ts{P}). ~x)$ and we have $\vdash t' : \Pi (\ts{P} : \El ~\o). ~\Prf ~\ts{P}$. It follows that $\tth$ is inconsistent. 
\end{proof}

\subsection{Examples of Interpretation}

We illustrate the interpretation with two examples. First, we detail the embedding of the theory of natural numbers into the theory of integers. This example has been implemented in \dk. Second, we give an informal presentation of the embedding of Zermelo set theory into a theory where sets are represented by graphs. These two examples exemplify the practicality and limitations of this interpretation.

\subsubsection{Natural Numbers and Integers}

We aim at interpreting the theory of natural numbers $\T_n$ in the theory of integers $\T_i$. We intuitively take $\ts{\nat} \coloneqq \int$. An integer is a non-negative natural number, so the predicate asserting that an integer is a natural number is defined by $\tp{\nat} \coloneqq \lambda z. ~z \geq_i 0_i$. The interpretation of $0_n$ is given by $\ts{0_n} \coloneqq 0_i$, and we choose $\tp{0_n} \coloneqq \ax_i^1 ~0_i$ for the proof of $\ts{0_n} \geq_i 0_i$. We take $\ts{\succ_n} \coloneqq \lambda \ts{x}. ~\lambda \tp{x}. ~\succ_i ~\ts{x}$ and $\ts{\succ_n} \coloneqq \lambda \ts{x}. ~\lambda \tp{x}. ~\ax_i^3 ~(\succ_i ~\ts{x}) ~\ts{x} ~0_i ~(\ax_i^2 ~\ts{x}) ~\tp{x}$. For the interpretation of $\geq_n$, we choose $\ts{\geq_n} \coloneqq \lambda \ts{x}. ~\lambda \tp{x}. ~\lambda \ts{y}. ~\lambda \tp{y}. ~\ts{x} \geq_i \ts{y}$. Given that $\geq_n$ returns a proposition, the parameter $\tp{\geq_n}$ must have type $\Pi \ts{x}. ~\Pi \tp{x}. ~\Pi \ts{y}. ~\Pi \tp{y}. ~\Prf ~(\ts{x} \geq_i \ts{y}) \ra \Prf ~(\ts{x} \geq_i \ts{y})$, which has an immediate inhabitant. The interpretation of $\ax_i^1$ is given by $\ts{(\ax_i^1)} \coloneqq \lambda \ts{x}. ~\lambda \tp{x}. ~\ax_i^1 ~\ts{x}$. Since $\ax_i^1$ returns a proof, and by definition of $\tp{\Prf}$, both $\ts{(\ax_i^1)}$ and $\tp{(\ax_i^1)}$ have the same type, so we can take $\tp{(\ax_i^1)} \coloneqq \ts{(\ax_i^1)}$. The parameters for $\ax_i^2$ and $\ax_i^3$ are chosen correspondingly. 

When defining the parameter $\ts{\rec_n}$, we assume $\ts{P}$ of type $\Pi (\ts{x} : \El ~\ts{\nat}). ~\Prf ~(\ts{x} \geq_i 0_i) \ra \El ~\o$. We must apply $\rec_i$ to a predicate of type $\El ~\ts{\nat} \ra \El ~\o$, which asserts that an integer $z$ is non-negative and that, given a proof $h_z$ of its non-negativity, it holds $\ts{P} ~z ~h_z$. Such a predicate can be encoded using $\fa$ and $\impd$. At some point in the proof, we want to show $\ts{P} ~z ~h_z$, but we can only derive $\ts{P} ~z ~h'_z$, where $h_z$ and $h'_z$ are two proofs of $z \geq_i 0_i$. To overcome this problem, we suppose \textit{proof irrelevance}
$$\mathsf{proof\_irr} : \Pi (p : \El ~\o) (h ~h' : \Prf ~p) (Q : \Prf p \ra \El ~\o). ~\Prf ~(Q ~h) \ra \Prf ~(Q ~h')$$ 
which states that two proofs of the same proposition are equal.

Using this interpretation of natural numbers into integers, we can derive for free the theorems of $\T_n$ in $\T_i$. For instance, we can show in $\T_n$ that $\vdash \thm : \Pi (x : \El ~\nat). ~\Prf ~(\succ_n ~x \geq_n 0_n)$, where $\thm$ is a proof that uses $\rec_n$, $\ax_n^1$, $\ax_n^2$ and $\ax_n^3$. The interpretation of $\T_n$ in $\T_i$ allows us to directly derive $\vdash \ts{\thm} : \Pi (\ts{x} : \El ~\int). ~\Prf ~(\ts{x} \geq_i 0_i) \ra \Prf ~(\succ_i ~\ts{x} \geq_i 0_i)$ in $\T_i$.

The complete interpretation of natural numbers into integers has been formalized in \dk, and is available in the file \texttt{nat\_sp.dk}.

\subsubsection{Sets and Pointed Graphs}

Sets can be represented by a more primitive notion of pointed graphs, such that this encoding satisfies Zermelo set theory~\cite{zermodulo}. Pointed graphs are directed graphs with a distinguished node---the root. In the \lpcm, pointed graphs are implemented~\cite{deduktiz} thanks to sorts $\graph$ and $\node$ of type $\Set$. The predicate $\et : \El\ \graph \ra \El\ \node \ra \El\ \node \ra \El\ \o$ is such that $\et ~a ~x ~y$ is the proposition asserting that there is an edge in pointed graph $a$ from node $y$ to node $x$. The operator $\root : \El\ \graph \ra \El\ \node$ returns the root of a pointed graph, and $\chgr : \El ~\graph \ra \El ~\node \ra \El ~\graph$ is such that $\chgr ~a ~x$ corresponds to the pointed graph $a$ in which the root is now at node $x$.

The different constructors on sets---unions, pairs, powersets and comprehension---are defined via rewrite rules using the structure of pointed graphs. At the end, every axiom of Zermelo set theory is a theorem in the theory of pointed graphs. Hence we can naturally interpret Zermelo set theory in the theory of pointed graphs. Remark that every pointed graph represents a set. It follows that the predicates asserting that an object of type $\El ~\graph$ is indeed a set are not necessary. 

The theory of pointed graphs is more computational than the usual Zermelo set theory. In particular, it satisfies a normalization theorem in deduction modulo theory~\cite{zermodulo}. Using such an interpretation, the theorems proved in Zermelo set theory can be transferred to the theory of pointed graphs.

\section{Conclusion}

In this paper, we have defined an interpretation of theories of the \lpcm with prelude encoding, given well-suited parameters for interpreting the constants of the source theory. If a source theory $\sth$ has an interpretation in a target theory $\tth$, then the theorems proved in $\sth$ come for free in $\tth$. At the end, we obtain a relative consistency result, establishing that the consistency of the theory $\tth$ entails the consistency of the theory $\sth$.

This interpretation applies when $\sth$ can be embedded into $\tth$. In particular, we allow the interpretation of a type $A$ of $\sth$ by a more general type $\ts{A}$ of $\tth$. As a consequence, we ensure that, for every term $t$ of type $A$ in $\sth$, its interpretation $\ts{t}$ of type $\ts{A}$ in $\tth$ indeed satisfies the predicate $\tp{A}$. Such an interpretation is well-suited when we embed a source theory into a more general target theory, as we have seen with natural numbers and integers. However, if the target theory encompasses exactly the source theory, then the translation introduces unnecessary predicates, as we have seen with sets and pointed graphs.

\paragraph{Practical application.} The \lpcm has been implemented in the \dk proof language and in the \lampi proof assistant. Future work would be to implement this interpretation in \dk. It would allow \textit{effective} proof transfers between different \dk theories, and would therefore strengthen the interoperability between proof assistants via \dk.

\paragraph{Theoretical application.} Dowek and Miquel~\cite{realizmod} developed a method for interpreting theories of first-order logic. They showed that this interpretation can be used to prove a relative normalization result for theories in deduction modulo theory~\cite{deduction_modulo}, that is first-order logic extended with user-defined rewrite rules. An application of this paper would be to prove a relative normalization result for the \lpcm. We would therefore be able to show that the encoding of set theory via pointed graphs in the \lpcm~\cite{deduktiz} satisfies a relative normalization result, just like this encoding in deduction modulo theory~\cite{zermodulo} does.

\section*{Acknowledgments}

The author is grateful to Valentin Blot, Gilles Dowek and Théo Winterhalter for their insightful feedback on this work, and thanks the reviewers for their relevant comments.

\nocite{*}
\bibliographystyle{eptcs}
\bibliography{biblio}
\end{document}